\documentclass[12pt]{ociamthesis}  % default square logo
\usepackage{amssymb}
\usepackage{amsthm}
\usepackage{amsmath}
\usepackage{algorithm}% http://ctan.org/pkg/algorithms
\usepackage{algpseudocode}% http://ctan.org/pkg/algorithmicx	
\usepackage{syntax}
\usepackage{listings}
\usepackage{bussproofs}
\usepackage{framed}
\usepackage{tikz}
\usepackage{pgfplots}
\usepackage{setspace}
\usepackage{hyperref}

\usepackage{hyperref}
\hypersetup{
    colorlinks,
    citecolor=black,
    filecolor=black,
    linkcolor=black,
    urlcolor=black
}
\usepackage[normalem]{ulem}

\usepackage[]{xcolor}

\title{Learning functional programs with function invention and reuse
    }   %note \\[1ex] is a line break in the title

\author{Andrei Diaconu}             %your name
\college{}

\degree{Honour School of Computer Science}     %the degree
\degreedate{Trinity 2020}         %the degree date

\begin{document}

%this baselineskip gives sufficient line spacing for an examiner to easily
%markup the thesis with comments
\baselineskip=18pt plus1pt

%set the number of sectioning levels that get number and appear in the contents
\setcounter{secnumdepth}{3}
\setcounter{tocdepth}{3}

\maketitle                  % create a title page from the preamble info
\begin{abstract}
\par Inductive programming (IP) is a field whose main goal is synthesising programs that respect a set of examples, given some form of background knowledge. This paper is concerned with a subfield of IP, inductive functional programming (IFP). We explore the idea of generating modular functional programs, and how those allow for function reuse, with the aim to reduce the size of the programs. We introduce two algorithms that attempt to solve the problem and explore type based pruning techniques in the context of modular programs. By experimenting with the implementation of one of those algorithms, we show reuse is important (if not crucial) for a variety of problems and distinguished two broad classes of programs that will generally benefit from function reuse.
\end{abstract}
          % include the abstract

\tableofcontents            % generate and include a table of contents

%now include the files of latex for each of the chapters etc
\newtheoremstyle{indented}
  {\topsep}{\topsep}%
  {}{}%
  {\bfseries}{}%
  {\newline}{}%
\theoremstyle{indented}
\newtheorem{exam}{Example}[chapter]

\chapter{Introduction}

\section{Inductive programming}

\indent Inductive programming (IP) \cite{gulwanietal} - also known as program synthesis or example based learning - is a field that lies at the intersection of several computer science topics (machine learning, artificial intelligence, algorithm design) and is a form of automatic programming. IP, as opposed to deductive programming \cite{deductive}  (another automatic programming approach, where one starts with a full specification of the target program) tackles the problem starting with an incomplete specification and tries to generalize that into a program. Usually, that incomplete specification is represented by examples, so we can informally define inductive programming to be the process of creating programs from examples using a limited amount of background information - we shall call this process the program synthesis problem \cite{shapiro}. We give an example of what an IP system might produce, given a task:

\begin{exam}
\textit{Input}: The definitions of \textit{map} and \textit{increment} and the examples $f([1,2,3]) = [2,3,4]$ and $f([5,6]) = [6,7])$.
\\
\textit{Output}: The definition $f = map$ $increment$.
\end{exam}

One of the key challenges of IP (and what makes it attractive) is the need to learn from small numbers of training examples, which mostly rules out statistical machine learning approaches, such as SVMS and neural networks. This can clearly create problems: if the examples are not representative enough, we might not get the program we expect.
\par As noted in the survey by Gulwani et al \cite{gulwanietal}, one of the main areas of research in IP has been end-user programming. More often than not, an application will be used by a non-programmer, and hence that user will probably not be able to write scripts that make interacting with that application easier. IP tries to offer a solution to that problem: the user could supply a (small) amount of information, such as a list of examples that describe the task, and an IP system could generate a small script that automates the task. Perhaps one of the most noteworthy applications of this idea is in the \textit{MS Excel} plug-in \textit{Flash Fill} \cite{gulwani2012spreadsheet}. Its task is to induce a program that generalizes some spreadsheet related operation, while only being given a few examples - usage of \textit{Flash Fill} can be seen in figure 1.1.

\begin{figure}
\begin{center}
\includegraphics{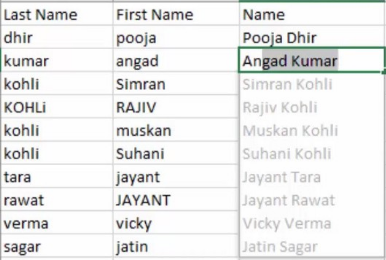}
\end{center}
\caption{Flash fill in action}
\end{figure}

\section{Motivation}

Two main areas of research in IP are inductive functional programming (IFP, which we will focus on in this paper) and inductive logic programming (ILP). The idea of function invention in the IFP context is not new, and indeed some systems use it, such as \textit{IGOR 2} and \textit{$\lambda^2$}. In informal terms, function invention mimics the way humans write programs: instead of writing a long one-line program, we break the bigger program into auxiliary functions that can be used to build a modular (equivalent) program.
\par In this context, we have asked the question of whether another program writing technique could be useful for inductive programming: reusing the functions that we have already invented. By reuse we mean that once a function has been invented, it can then be used in the definition of another function. While some ILP systems have explored the idea of reusing functions (such as \textit{Metagol} and \textit{Hexmil} \cite{metagol} and to a lesser extent \textit{DILP} \cite{dilp} and \textit{ILASP} \cite{ilasp}), function reuse and its benefits (if any) have not really been explored in the IFP context, as noted by Cropper \cite{cropperth}. When investigating the existing systems with invention capabilities, we have observed that the way the invention process is conducted makes reuse practically impossible. Moreover, even though predicate invention and reuse have been claimed as useful (at least in the ILP context \cite{turns20}), to our knowledge there has been no work that empirically demonstrates that it is, nor any work discussing when it may be useful.
To address those limitations, in this work we are interesting in the following research questions:

\begin{itemize}
\item[\textbf{Q1}] Can function reuse improve learning performance (find programs faster)?
\item[\textbf{Q2}] What impact does modularity have on pruning techniques, especially type based ones?
\item[\textbf{Q3}] What impact does the grammar of the synthesized programs have on function reuse?
\item[\textbf{Q4}] What classes of problems benefit from it; that is, can we describe the kinds of programs where function reuse is useful?
\end{itemize}

\section{Contributions}
\indent \indent In this paper, we make the following contributions:
\begin{itemize}
\item We provide a formal framework to describe IFP approaches that solve the program synthesis problem by creating modular programs and that can exploit function reuse.
\item Given this formal framework, we create two algorithms that solve the synthesis problem. One of them uses type based pruning to speed up the searching process, but uses a restrictive grammar; we have proven that for general grammars, this algorithm (which uses a ``natural" type inference based pruning approach) loses completeness (which in particular greatly hinders function reuse). The second algorithm does not use type based pruning and works with general grammars, but we propose a way in which this might be achieved.
\item Our experimental work has provided positive results, which shed light on the usefulness of reuse in the IFP context; for example, we have shown that reuse can decrease the size of the synthesized programs and hence reduce the overall computation time (in some cases dramatically). Through experimentation, we have also distinguished two classes of problems for which reuse is important: AI planning problems and problems concerned with nested data structures (we have focused on lists).
\end{itemize}

\section{Structure of the report}
\indent \indent The rest of the paper is structured as follows:
\begin{itemize}
\item \textbf{chapter 2}: Presents background on inductive programming, function invention and reuse, and a variety of other systems.
\item \textbf{chapter 3}: Presents a formal framework for describing the program synthesis problem and formalizes function reuse.
\item \textbf{chapter 4}: Presents two algorithms that attempt to solve the program synthesis problem, in light of the description presented in chapter 3.
\item \textbf{chapter 5}: Explores the role of function reuse through experimentation and contains a variety of experiments that validate our hypothesis; we also explore the various use cases of function reuse.
\item \textbf{chapter 6}: Presents the conclusions, limitations, possible extensions of the project.
\end{itemize}

\chapter{Background and related work}

\indent In the previous chapter, we have informally introduced the concept of inductive programming (IP), presented its relevance and showcased our ideas. In this chapter, we first provide the reader with more background on IP (areas of research, approaches) and then switch to literature review, showing different IP systems and their relevance to ours. We finish the chapter by talking about the idea of function invention and reuse.

\section{Background on IP}
IP has been around for almost half a century, with a lot of systems trying to tackle the problem of finding programs from examples. It is a subject that is placed at the crossroad between cognitive sciences, artificial intelligence, algorithm design and software development \cite{kitzelmannsurvey}. An interesting fact to note is that IP is a machine learning problem (learning from data) and moreover, in recent years it has gained attention because of the inherent transparency of its approach to learning, as opposed to the black box nature of statistical/neuronal approaches, as noted by Schmid \cite{SchmidInductivePA}.
\par IP has two main research areas, as noted by Gulwani \textit{et al.} \cite{gulwanietal}:
\begin{itemize}
\item Inductive functional programming (IFP): IFP focuses on the synthesis of functional programs, typically used to create programs that manipulate data structures.
\item Inductive logical programming (ILP): ILP started as research on induction in a logical context \cite{gulwanietal}, generally used for learning AI tasks. It's aim is to construct a hypothesis (logic programs) \textit{h} which explain examples \textit{E} in terms of some background knowledge \textit{B} \cite{MUGGLETON1999283}.
\end{itemize}
\par As highlighted in the review by Kitzelmann \cite{kitzelmannsurvey}, there have been two main approaches to inductive programming (for both IFP and ILP):
\begin{itemize}
\item \textbf{analytical approach}: Its aim is to exploit features in the input-output examples; the first systematic attempt was done by Summers' \textit{THESIS} \cite{thesis} system in 1977. He observed that using a few basic primitives and a fixed program grammar, a restricted class of recursive LISP programs that satisfy a set of input-output examples can be induced. Because of the inherent restrictiveness of the primitives, the analytical approach saw little innovation in the following decade, but systems like \textit{IGOR1}, \textit{IGOR2} \cite{igor2} have built on Summers' work. The analytical approach is also found in ILP, a well known example being Muggleton's \textit{Progol} \cite{progol}.
\item \textbf{generate-and-test approach (GAT)}: In GAT, examples are not used to actually construct the programs, but rather to test streams of possible programs, selected on some criteria from the program space. Compared to the analytical approach, GAT tends to be the more expressive approach, at the cost of higher computational time. Indeed, the \textit{ADATE} system, a GAT system that uses genetic programming techniques to create programs, is one of the most powerful IP system with regards to expressivity \cite{kitzelmannsurvey}. Another well known GAT system is Katayama's \textit{Magic Haskeller} \cite{mhask}, which uses type directed search and higher-order functions as background knowledge. Usually, to compensate for the fact that the program space is very big, most GAT systems will include some sort of pruning that discards undesirable programs.
\end{itemize}

\section{Related work}
\par \indent We now present three systems that helped us develop our ideas and contrast them with our work.
\subsection{Metagol}
\textit{Metagol} \cite{metagol} is an ILP system that induces Prolog programs. It uses an idea called MIL, or meta-interpretative learning, to learn logic programs from examples. It uses three forms of background information:
\begin{itemize}
\item compiled background knowledge (CBK): those are small, first order Prolog programs that are deductively proven by the Prolog interpreter.
\item interpreted background knowledge (IBK): this is represented by higher-order formulas that are proven with the aid of a meta-interpreter (since Prolog does not allow clauses with higher-order predicates as variables); for example, we could describe \textit{map/3} using the following two clauses:\\ $map([],[],F) :- $ and $map([A|As],[B|Bs],F) :- F(A,B), map(As,Bs,F)$.
\item metarules: those are rules that enforce the form (grammar) of the induced program's clauses; an example would be $P(a, b) :- Q(a, c), R(c, b)$, where upper case letters are existentially quantified variables (they will be replaced with CBK or IBK).
\end{itemize}
\par The way the hypothesis search works is as follows: try to prove the required atom using CBK; if that fails, fetch a metarule, and try to fill in the existentially quantified variables; continue until a valid hypothesis (one that satisfies the examples) is found. Something to note here is that \textit{Metagol} generates new examples: if we select the \textit{map} metarule, based on the existing examples we can infer a set of derived examples that the functional argument of \textit{map} must satisfy. This technique is used to prune incorrect programs from an early stage. All this process is wrapped in a depth-bounded search, so as to ensure the shortest program is found.
\par Our paper has started as an experiment to see whether ideas from \textit{Metagol} could be transferred to a functional setting; hence, in the next chapters we use similar terminology, especially around metarules and background knowledge. We will also use depth-bounded search in our algorithm, for similar reasons to \textit{Metagol}.

\subsection{Magic Haskeller}
Katayama's \textit{Magic Haskeller} \cite{mhask} is a GAT approach that uses type pruning and exhaustive search over the program space. Katayama argues that type pruning makes the search space manageable. One of the main innovation of the system was the usage of higher-order functions, which speeds up the searching process and helps simplify the output programs (which are chains of function applications). Our system differs in the fact that our programs are modular, which allow for function reuse. One of \textit{Magic Haskeller}'s limitations is the inability to provide user supplied background knowledge. The implementation of our algorithms enable a user to experiment with the background functions in a programmatic manner, and we also make it fairly easy to change the grammar of the programs.
\subsection{$\lambda^{2}$}
\textit{$\lambda^{2}$} \cite{lambdasq} is an IFP system which combines GAT and analytical methods: the search is similar to \textit{Magic Haskeller}, in the way that it uses higher order functions and explores the program space using type pruning, but differs in the fact that programs have a nested form (think of \textit{where} clauses in Haskell) and uses an example propagation pruning method, similar to \textit{Metagol}. However, such an approach does not allow function reuse, since an inner function can't use an ``ancestor" function in its definition (possible infinite loop). Our paper tries to address this, exploring the possibility of creating non-nested programs and hence allowing function reuse.

\section{Invention and reuse}

Generally, most IP approaches tend to disregard the extra knowledge found during the synthesis process as another form of background knowledge. In fact, systems like \textit{$\lambda^{2}$} and \textit{Magic Haskeller} make this impossible because of how the search is conducted. Some systems, like \textit{Igor 2} do have a limited form of it, but it is very restrictive and does not allow function reuse in a general sense. This usually stems from what grammars the induced programs use. One of our main interests has been the usefulness of function reuse by allowing a modular (through function invention) way of generating programs (that is, we create ``standalone" functions that can then be pieced together like a puzzle). For example, consider the \textit{drop lasts} problem: given a non-empty list of lists, remove the last element of the outer list as well as the last elements of all the inner ones. Example 2.1 shows a possible program that was synthesized using only invention. However, if function reuse is enabled, example 2.2 shows how we can synthesize a simpler program, which we would expect to reduce the searching time.

\begin{exam}[\textit{droplasts} - only invention]
\begin{lstlisting}[language=Haskell]
target = f1.f2
f1 = f2.f4
f2 = reverse.f3
f3 = map reverse
f4 = tail.f5
f5 = map tail
\end{lstlisting}
\end{exam}

\begin{exam}[\textit{droplasts} - invention + reuse]
Note how $f1$ is reused to create a shorter program.
\begin{lstlisting}[language=Haskell]
target = f1.f2
f2 = map f1
f1 = reverse.f3
f3 = tail.reverse
\end{lstlisting}
\end{exam}

An interesting questions when considering function reuse is what kind of programs benefit from it, which we explore in chapter 5, but we will now move to formalizing the program synthesis problem.

\theoremstyle{definition}
\newtheorem{defn}{Definition}[chapter]

\theoremstyle{plain}
\newtheorem{claim}{Claim}
\newtheorem{thm}{Theorem}[chapter]

\chapter{Problem description}
Before presenting algorithms that solve the synthesis problem, we need to formalize it. We will assume, for the rest of the chapter, that all definitions will be relative to a target language $\mathcal{L}$, whose syntax will be specified in the next chapter.

\section{Abstract description of the problem}

A program synthesis algorithm's aim is to induce programs with respect to some sort of user provided specification. The synthesis process will create programs which we call \textit{induced programs}, that are composed of a set of functions which we call \textit{induced functions}. For each induced program we will distinguish a function called the target function, which is to be applied to the examples to check whether a candidate program is a solution. Intuitively, the output shall be an induced program whose target function satisfies the provided specification.
\par The provided specification in this paper shall be divided in two parts: background knowledge and input-output examples.

\begin{defn}[Background knowledge (BK)]
We define background knowledge to be the information used during the synthesis process. The BK completely determines the possible forms an induced program can have. There are three types of BK that we consider:
\begin{itemize}
	\item \textit{Background functions}: represents the set of functions provided via an external source. We require those functions to be total so as to not introduce non-termination in an induced program. We use the notation $BK_F$ to refer to this kind of knowledge.
	\item \textit{Invented functions}: represents the set of functions that are invented during the synthesis process; this set grows dynamically during the synthesis process (with each new invented function). We use the notation $BK_I$ to refer to this kind of knowledge.
	\item \textit{Function templates}: a set of lambda calculus-style contexts that describe the possible forms of the induced functions. We use the notation $BK_T$ to refer to it.
\end{itemize}
\end{defn}

\par Let us unpack this definition. We have referred to both $BK_I$ and $BK_F$ to be sets of functions: more precisely, they are sets containing pairs of the form $(f_{name}, f_{body})$: $f_{name}$ represents the identifier (function name) of function $f$, whereas $f_{body}$ corresponds to the body of its definition. When we write $f$, we refer to the pair. Example 3.1 shows an example of two functions that might be part of $BK_F$.

\begin{exam}[Background functions]
		\begin{lstlisting}[language=Haskell]
rev xs = if xs == []
         then []
         else rev (tail xs) ++ [head xs]
addOne x = x + 1
		\end{lstlisting}
\end{exam}

\par Function templates represent the blueprints that we use when defining invented functions. They are contexts in the meta-theory of lambda calculus sense, that represent the application of a higher-order combinator, where the ``holes" are place-holders for the required number of functional inputs for such a combinator. Those place-holders signify that there is missing information in the body of the induced function that needs to be filled with either background functions or synthesized functions. We have chosen those to specify the grammar of the invented functions because higher-order functions let us combine the background and invented functions in complex ways, and provide great flexibility. We note the similarity of our function templates to \textit{metarules} \cite{metarules} and \textit{sketches} \cite{sketches}, which serve similar purposes in the respective systems where they are used. Example 3.2 shows the form of a few such templates. For convenience, we number the ``holes", e.g. $\fbox{i}$, with indices starting from 1.

\begin{exam}[Function templates]
		\indent \indent Conditional templates: \textbf{if} $\fbox{1}$ \textbf{then} $\fbox{2}$ \textbf{else} $\fbox{3}$. \\
		\indent Identity: $\fbox{1}$. \\
		\indent Higher-order templates: $\fbox{1}$ \textbf{.} $\fbox{2}$, \textbf{map} $\fbox{1}$, \textbf{filter} $\fbox{1}$. \\
\end{exam}

\par We say an induced function is \textit{complete} if its body has no holes and all functions mentioned in it have a definition, or \textit{incomplete} otherwise. Similarly, we say an induced program is \textit{complete} if it is composed of complete functions. We give a short example to see how templates and functions interact with each other, which will provide some intuition for the algorithmic approach to the inductive process presented in the next chapter.

\begin{exam}[Derivation]
Suppose we wish to find the complete function $F=$ map $reverse$. The following process involving the BK will take place: we invent a new function $F$, and assign it the \textit{map} template to obtain the definition $F=$ map $\fbox{1}$; we then fill the remaining hole using  \textit{reverse}.
\end{exam}

\par The second part of the specification is represented by examples.
\begin{defn}[Examples]
Examples are user provided input-output pairs that describe the aim of the induced program.
We shall consider two types of examples:
\begin{itemize}
	\item Positive examples: those specify what an induced program should produce;
	\item Negative examples: those specify what an induced program should \emph{not} produce;
\end{itemize}
\end{defn}
\par We use the relation $in \mapsto^{+} out$ to refer to positive examples, and the relation $in \mapsto^{-} out$ to refer to negative ones.
While the positive examples have a clear role in the synthesis process, the negative ones serve a slightly different one: they try to remove ambiguity, which is highlighted in example 3.4. Something to note is that both the positive and the negative examples need to have compatible types, meaning that if a type checker is used, all the inputs would share the same type, and so should the outputs.

\begin{exam}[Examples]
Given the positive examples $[3,2,1] \mapsto^{+} [1,2,3]$ and $[5, 4] \mapsto^{+} [4, 5]$, and the negative examples $[1,3,2] \mapsto^{-} [2,3,1]$ and $[5,4] \mapsto^{-} [5,4]$ then the program we want to induce is likely to be a list sorting algorithm. Note that if we only look at the positive examples, another possible induced program is \textit{reverse}, but the negative example $[1,3,2] \mapsto^{-} [2,3,1]$ removes this possibility.
\end{exam}

\begin{defn}[Satisfiability]
We say an complete induced program $P$ whose target function is $f$ satisfies the relations $\mapsto^{+}$ and $\mapsto^{-}$ if:
\begin{itemize}
\item $\forall (in, out) \in \mapsto^{+} . f(in) = out$
\item $\forall (in, out) \in \mapsto^{-} . f(in) \neq out$
\end{itemize}
\end{defn}

\begin{defn}[Program space]
Assume we are given background knowledge $BK$ and let $\mathcal{T}_{check}$ be a type checking algorithm for $\mathcal{L}$. We define the program space $\mathcal{P}_{BK}$ to be composed of programs written in $\mathcal{L}$ such that:
\begin{itemize}
   \item[1.] the bodies of the induced functions are either function templates (which still have holes in them) or function applications (for completed functions).
   \item[2.] the inputs for the higher-order functions (of the templates) are either functions from $BK_I$ or $BK_F$.
   \item[3.] they are typeable w.r.t. $\mathcal{T}_{check}$.
   \item[4.] they contain no cyclical definitions (guard against non-termination).
\end{itemize}
\end{defn}

Note how the $\mathcal{P}_{BK}$ contains induced programs whose functions could still have unfilled holes. We now describe the solution space, which contains the programs we consider solutions.

\begin{defn}[Solution space]
Given $BK$, $\mapsto^+$ and $\mapsto^-$, we define the solution space $\mathcal{S}_{BK, \rightarrow^+, \rightarrow^-} \subset \mathcal{P}_{BK}$ to be composed of complete programs whose target functions satisfy both $\rightarrow^+$ and $\rightarrow^-$.
\end{defn}

We now formulate the program synthesis problem.
\begin{defn}[Program Synthesis problem]
Given:

\begin{itemize}
    \item a set of positive input/output examples,
    \item a set of negative input/output examples,
    \item background knowledge $BK$
\end{itemize}

\noindent
the \emph{Program Synthesis problem} is to find a solution $S \in \mathcal{S}_{BK, \mapsto^+, \mapsto^-}$ that has the minimum number of functions (textually optimal solution).
\noindent
\end{defn}

\section{Invention and reuse}
For this section, suppose $\mathcal{A}$ is an algorithm that solves the \textit{Program Synthesis} problem. First we formalize the concepts of invention and reuse, which we mentioned in chapters 1 and 2.
\begin{defn}[Invention]
We say that $\mathcal{A}$ can invent functions if at any point during the synthesis process it is able to fill a hole with a fresh function name (i.e. does not appear in any of the previous definitions).
\end{defn}

\begin{defn}[Reuse]
We say that $\mathcal{A}$ can reuse functions if at any point during the synthesis process it is able to fill a hole with a function name that has been invented at some point during the synthesis process (be it defined or yet undefined).
\end{defn}

\par As we can see, the two definitions are intertwined: we can not have reuse without invention. The motivation for inventing functions is that this creates modular programs, which naturally support function reuse. As we shall see in the next chapter, one of the main consequences with modularity is its effect on type based pruning techniques. 
\par When function reuse is used, certain problems will benefit from this (such as \textit{droplasts} from chapter 2): we could find solutions closer to the root, which can have noticeable effects on the computation time. However, enabling function reuse means that the BK increases with each newly invented function, and hence the branching factor of the search tree increases dynamically: in the end, function reuse can be seen as a trade-off between the depth and the branching factor of the search tree; this will benefit some sorts of problems, but for others it will just increase the computation time. The concerns we talked in this paragraph are related to the research questions posed in section 1.2, which we will address in the next two chapters.

\newtheorem{theorem}{Theorem}[chapter]
\newtheorem{lemma}{Lemma}[chapter]
\lstset{
	language=Haskell,
	breaklines=true
}
\renewcommand{\algorithmicrequire}{\textbf{Input:}}
\renewcommand{\algorithmicensure}{\textbf{Output:}}
\makeatletter
 \renewcommand{\ALG@name}{Pseudocode}
\makeatother

\chapter{Algorithms for the program synthesis problem}
As previously presented, we want to create an algorithm that is able to create modular programs and hence able to reuse already induced functions. Two of the main aims of such algorithms should be soundness and completeness, which we define next (assume $BK$, $\mapsto^-$ and $\mapsto^+$ are given). 

\begin{defn}[Completeness]
We say an algorithm that solves the program synthesis problem is complete if it is able to synthesize any program in $\mathcal{S}_{BK, \mapsto^+, \mapsto^-}$.
\end{defn}

\begin{defn}[Soundness]
We say an algorithm that solves the program synthesis problem is sound if the complete programs it synthesizes have their target function satisfy $\mapsto^+$ and $\mapsto^+$.
\end{defn}

\par Motivated by \textbf{Q2} from section 1.2, we are interested in another property of such algorithms, namely \emph{type pruning}: we wish to discard undesirable programs based on their types (e.g. when they become untypable). As we shall see, this third property will lead to the creation of two algorithms, but next we present some preliminaries required for understanding them.

\section{Preliminaries}
\subsection{Target language and the type system}
We have chosen the target language to be a $\lambda$-like language that supports contexts, since we don't want to introduce too much added complexity, while still having enough expressivity. Its syntax can be seen in figure 4.1. For simplicity, when we provide code snippets in this language we will adopt a slightly simpler (but equivalent) notation: for example, $val \, f = map \,( \, \textit{Variable} \, g \,)$ will be written as $f = map \, g$. The language supports both recursive and non-recursive definitions for the background knowledge. It also has a number primitives such as $+$, $==$, \textit{nil} or \textit{(:)} (the last two let us work with lists).
\par To support type based pruning, our language will be fully typed, the typing rules being shown in figure 4.2 (for brevity we have omitted the typing rules for primitives, apart for \textit{nil} and \textit{(:)}). We define some standard type theoretic terms we will use later:
\begin{itemize}
\item typing environment: Usually denoted using $\Gamma$, it represents a map that associates a type (quantified or not) to a name or a type variable.
\item substitution: This represent a way to replace existing type variables in an unquantified type with other types; those can also be applied to typing environments (i.e. mapping the application of the substitution over all the types in the environment).
\item type inference: The process of inferring the type of an expression, given some typing environment.
\item free variable: The free variables of a type or environment are those typing variables which are not bound; we use the notation $\mathit{ftv}(...)$ to denote the set of free variables.
\item generalizing: We can generalize over a type by closing over the free variables of that type.
\item instantiating: We can instantiate a quantified type by replacing all the bound variables in it with fresh type variables (and hence make it unquantified).
\item unification: Given two types, unification is the process that yields a substitution that when applied to the types makes them equal; we will use the function $\mathit{unify}$ to denote this process (which can fail; when we write $\mathit{unify}(\tau_1, \tau_2)$ in a condition, we implicitly mean "if they can unify").
\end{itemize}

\begin{figure}
\begin{framed}
\begin{grammar}
<decl> ::= `val' <ident> `=' <expr>  -- \textit{non recursive definition}
\alt `rec' <ident> `=' <expr> -- \textit{recursive definition}
\alt `Pex' <expr> $\Rightarrow$ <expr> -- \textit{a way to specify positive examples}
\alt `Nex' <expr> $\Rightarrow$ <expr>-- \textit{a way to specify negative examples}

<expr> ::= `Num' n
\alt `Char' c
\alt `True'
\alt `False'
\alt `Variable' <ident>
\alt `Lambda' [<ident>] <expr>
\alt <expr> <expr>
\alt `If' <expr> `then' <expr> `else' <expr>'
\alt \fbox{i} -- \textit{we need to represent holes in the syntax}
\end{grammar}
\end{framed}

\caption{BNF Syntax}
\end{figure}

\begin{figure}
\begin{framed}
\begin{equation*}
  \AxiomC{$n \in \mathbb{Z}$}
  \RightLabel{(TNum)}
  \UnaryInfC{$\Gamma \vdash Num \, n: Integer$}
  \DisplayProof
\end{equation*}
\\
\begin{equation*}
  \AxiomC{$c \in \{'a', ... 'z', 'A' ... 'Z', '0', ..., '9'\}$}
  \RightLabel{(TChar)}
  \UnaryInfC{$\Gamma \vdash Char \, c: Character$}
  \DisplayProof
\end{equation*}
\\
\begin{equation*}
  \AxiomC{-}
  \RightLabel{(TTrue)}
  \UnaryInfC{$\Gamma \vdash True: Boolean$}
  \DisplayProof
  \qquad
   \AxiomC{-}
  \RightLabel{(TFalse)}
  \UnaryInfC{$\Gamma \vdash False: Boolean$}
  \DisplayProof
\end{equation*}
\\
\begin{equation*}
  \AxiomC{$x:\tau \in \Gamma$}
  \RightLabel{(TVar)}
  \UnaryInfC{$\Gamma \vdash $Variable $ x : \tau$}
  \DisplayProof
  \qquad
  \AxiomC{$\Gamma \vdash e_0 : \tau \rightarrow \tau'$}
  \AxiomC{$\Gamma \vdash e_1: \tau$}
  \RightLabel{(TApp)}
  \BinaryInfC{$\Gamma \vdash e_0 \, e_1: \tau'$}
  \DisplayProof
\end{equation*}
\\
\begin{equation*}
  \AxiomC{$\Gamma\vdash e : Boolean$}
  \AxiomC{$\Gamma\vdash e' : \tau$}
  \AxiomC{$\Gamma\vdash e'': \tau$}
  \RightLabel{(TIf)}
  \TrinaryInfC{$\Gamma \vdash$ If $e$ then $e'$ else $e'' : \tau$}
  \DisplayProof
\end{equation*}
\\
\begin{equation*}
  \AxiomC{-}
  \RightLabel{(TNil)}
  \UnaryInfC{$\Gamma \vdash nil : \forall \alpha \, . \, [\alpha]$}
  \DisplayProof
  \qquad
  \AxiomC{$\Gamma \vdash e_1 : \tau$}
  \AxiomC{$\Gamma \vdash e_2 : [\tau]$}
  \RightLabel{(TCons)}
  \BinaryInfC{$\Gamma \vdash (e_1 : e_2) : [\tau]$}
  \DisplayProof
\end{equation*}
\\
\begin{equation*}
  \AxiomC{$\Gamma, x:\tau \vdash e : \tau'$}
  \RightLabel{(TAbs)}
  \UnaryInfC{$\Gamma \vdash \lambda x \,.\, e : \tau \rightarrow \tau'$}
  \DisplayProof
  \qquad
  \AxiomC{$\Gamma\vdash e : \tau$}
  \AxiomC{$\bar{\alpha} \notin ftv(\Gamma)$}
  \RightLabel{(TGen)}
  \BinaryInfC{$\Gamma \vdash  e : \forall \bar{\alpha} \, . \,\tau$}
  \DisplayProof
\end{equation*}
\\
\begin{equation*}
  \AxiomC{$\Gamma\vdash e : \tau$}
  \AxiomC{$\alpha$ is  an instantiation of $\tau$}
  \RightLabel{(TInst)}
  \BinaryInfC{$\Gamma \vdash e : \alpha$}
  \DisplayProof
\end{equation*}
\\
\begin{equation*}
  \AxiomC{ $\alpha \not\in ftv(\Gamma)$ }
  \RightLabel{(THole)}
  \UnaryInfC{$\Gamma \vdash \fbox{i} : \forall \alpha \, . \, \alpha$}
  \DisplayProof
\end{equation*}

\end{framed}
\caption{Typing rules for \textit{expr}}
\end{figure}

\subsection{Combinatorial search}
Most systems that have a generate-and-test or hybrid approach use some form of combinatorial search as a means to find new programs. \textit{Metagol} uses iterative deepening depth first search (IDDFS), \textit{MagicHaskeller} uses BFS etc. It is natural then that we shall also use a form of combinatorial search, more specifically we use IDDFS. We have chosen this approach because we want to synthesize programs that gradually increase in size, so as to ensure that the induced program will be the shortest one in terms of the number of synthesized functions.

\subsection{Relations concerning the program space}
We first formalize what was meant by ``cyclical definitions" in the definition of the program space.

\begin{defn}[Name usage]
We say that a function $f$ \textit{directly uses} another function $f'$ if $f'_{\mathit{name}}$ appears in $f_{\mathit{body}}$. \\
We say that a function $f$ \textit{uses} another function $f'$ if $f$ \textit{directly uses} $f'$ or $f$ \textit{uses} $g$ and $g$ \textit{directly uses} $f'$ (transitive closure of \textit{directly uses}).\\
We say that a program is acyclic if no function in $P$ \textit{uses} itself, or cyclic otherwise.
\end{defn}
\par Intuitively, if a complete program is cyclic, then it might not terminate (although this is not always the case, e.g. programs that include co-recursive functions), so we want to avoid such programs to not introduce non termination.

\begin{defn}[Specialisation step]
Given a function $f$, a typing environment $\Gamma$, the specialisation step is a relation on $\mathit{Expr} \times (\mathit{Expr}, \mathit{TypingEnvironment})$ indexed by a typing environment defined by the following rule:
\begin{itemize}
\item if
\begin{itemize}
\item[1.] a hole $ \fbox{i}$ appears in $f_{body}$;
\item[2.] $ g \in BK_I \, \cup \, BK_F$ and $g$ does not use $f$ OR $g$ is a fresh function we invent and add to $BK_I$, whose type in $\Gamma$ is a fresh type variable (use an existing function or invent);
\item[3.] $\rho = \mathit{unify}(\tau, \Gamma(g_{name}))$, where $\tau$ is the type inferred for $\fbox{i}$ in the type derivation tree for $f_{body}$;
\end{itemize}
then we write $f_{body} \xrightarrow{\Gamma} (f_{body}[\, g_{name} \, / \, \fbox{i}\,], \rho \, \Gamma)$ (note, $\rho \, \Gamma$ means we apply the substitution $\rho$ to $\Gamma$).
\end{itemize}
\end{defn}

\par We briefly give some intuition on why the environment is updated. What we do in step 3 is we try to mimic the \textit{TApp} typing rule. Since the type of the hole we are to fill represents a ``minimum requirement" for what types can fill it (see the \textit{THole} rule), it suffices to make sure the type of the filler function is unifiable with the type of the hole. Now, because of our top-down approach to typing, we keep the types of the invented functions unquantified, since the types of those functions are intertwined. Hence, changes we make in one can lead to changes in other functions' types, so we need to apply the substitution to the whole environment (note that both the background functions and the higher-order combinators are quantified, so this won't affect them).
\par Based on this relation, we give an ordering on pairs of induced programs and their typing environment (we will call those pairs \emph{program states}).

\begin{defn}[Ordering]
We say that a program state $(P, \Gamma)$ is more concrete than another program state $(P', \Gamma')$ if either
\begin{itemize}

\item (\textit{specialize}) $names(P) = names(P')$ and there exist exactly two functions $f \in P$ and $f' \in P'$ such that
	\begin{itemize}
		\item[1.] $f_{name} = f'_{name}$
		\item[2.] $f_{body} \xrightarrow{\Gamma} (f'_{body}, \Gamma')$
	\end{itemize}

\item or (\textit{define}) $P' = P \cup \{f_{name} = T\}$, with
\begin{itemize}
\item[1.] $f$ \textit{used} by another function but not yet defined in $P$
\item[2.] $T \in BK_T$
\item[3.] $\rho = \mathit{unify}(\tau, \Gamma(f_{name}))$ (where $\tau$ is the type that can be inferred for the template $T$) and $\Gamma' = \rho \, \Gamma$.
\end{itemize}
\end{itemize}
We write this as $(P, \Gamma) \preccurlyeq (P', \Gamma')$.
\end{defn}
Here, step 3 in the \textit{define} rule makes sure that the type of the template ``agrees" with all the constraints collected from other previous uses of the to be defined function (and again we need to update the whole environment).

Let $\preccurlyeq^{*}$ be the reflexive, transitive closure of $\preccurlyeq$, defined by:
\begin{itemize}
\item $P = P' \Rightarrow P \preccurlyeq^{*} P'  $
\item $\exists P''. P \preccurlyeq^{*} P" \land P" \preccurlyeq P' \Rightarrow  P \preccurlyeq^{*} P'  $.
\end{itemize}

\subsection{Partitioning the templates}
Before continuing with the algorithm, we will divide the templates in two categories, based on their type signatures.

\begin{defn}[Linear function templates]
We define a \textit{linear function template} to be a template where the types of the functional inputs of the associated higher-order function (the function that is applied in the template) share no type variables.
\end{defn}

\begin{exam}[Linear function templates]
The \textit{map} template is a linear function template. The type of the associated combinator is $\forall ab . (a \rightarrow b) \rightarrow [a] \rightarrow [b]$, so the only input (the function we are mapping) has type $a \rightarrow b$, hence it is trivially a linear function template.
\end{exam}

\begin{defn}[Branching function templates]
We define a \textit{branching function template} to be a template where the types of the functional inputs of the associated higher-order function share type variables.
\end{defn}

\begin{exam}[Branching function templates]
The \textit{composition} template is a branching function template. Its type is $\forall abc . (b \rightarrow c) \rightarrow (a \rightarrow b) \rightarrow a \rightarrow c$, so we have two functional inputs to the template, whose types share a type variable, namely $b$.
\end{exam}

\section{The algorithms}
We will now present two algorithms that solve the program synthesis problem and focus on synthesizing \emph{modular programs} that allow \emph{function reuse} to be employed.
\par The reason we show two such algorithms is because we wish to explore \textbf{Q2} from section 1.2 here:
\begin{itemize}
\item we will show that when only linear templates are considered, there exists a sound and complete algorithm that is capable of effective type pruning.
\item we show how the addition of branching function templates breaks the ``naive" algorithm's completeness, which highlights that ``natural" type pruning techniques do not transpose well when general background knowledge is used (and hence it is not a trivial task).
\end{itemize}

\subsection{Only linear function templates algorithm}
We begin with the algorithm that uses only linear function templates. We assume that the $BK$ for the remainder of the subsection will only contain linear templates. Also assume we are given the relations $\mapsto^+$ and $\mapsto^-$.
\par The core parts of the algorithm are 3 procedures that work together: \textit{programSearch}, \textit{expand} and \textit{check}. \textit{programSearch} does the actual search and represents the entry point for the algorithm; \textit{expand} and \textit{check} are both used in \textit{programSearch} to generate new programs and to check whether a certain program is a solution, respectively.
\par We now describe each one in more detail:
\begin{itemize}
\item \textit{programSearch}: this procedure does a depth bounded search; it takes 3 inputs: a procedure that expands nodes in the search tree (\textit{expand}), a function that checks whether the current node is a solution (\textit{check}) and an initial program state. The initial state is defined to be $(P_{empty}, \Gamma_{empty})$, where $P_{empty}$ is the empty set and $\Gamma_{empty}$ is the typing environment that contains the (\emph{quantified}) types of the higher-order functions used in templates and the background functions. The procedure uses an auxiliary function, \textit{boundedSearch}, which actually does the search by expanding and checking nodes in a similar fashion to DFS, but the search is conducted to a certain depth (cut-off point), which is gradually increased. Pseudocode for this procedure can be seen in \textit{Pseudocode 1}.

\item \textit{expand}: given a program state $(P, \Gamma)$, this procedure implements one use of the $\preccurlyeq$ relation to create a stream of more concrete programs. Note that because of the typing conditions in the rules of $\preccurlyeq$ it is here that we are pruning untypable programs. To make this procedure deterministic, we must specify how the two rules from $\preccurlyeq$ are used to create the stream: if $P$ contains a function that has at least a hole, pick that hole and fill it using the \textit{specialize} rule. If no hole remains, use the \textit{define} rule to define one of the previously invented but undefined functions (the target function is the first invented function). The reason behind this strategy is that it makes sure we are filling the holes as soon as possible, and hence detect untypable programs early. Of course, when applying the \textit{specialize} rule, we will try to fill the hole in all the possible ways (by either using previously invented or background functions OR by using a freshly invented function), to ensure determinism; similarly, when the \textit{define} rule is used, we will try to assign all the possible templates to the function we want to define.

\item \textit{check}: checks whether a given node (a program state) is a solution. First, we make sure that the type of the candidate program's target is compatible (unifiable) with the type of the examples. If it is, we then check whether the program satisfies the examples using an interpreter for the language (using a program environment that contains the definitions of the background functions and the higher-order functions, to which we add the invented functions); pseudocode for this can be found in \textit{Pseudocode 2}.
\end{itemize}

\begin{algorithm}
\caption{progSearch}
\begin{algorithmic}
\Procedure{progSearch}{$expand, check, init$}
\Ensure An induced program
\For{$depth=1$ to $\infty$}
	\State $result \gets boundedSrc(expand, check, init, depth)$
	\If{$result \neq failure$}
    		\State \textbf{return} $result$
    \EndIf
\EndFor
\EndProcedure
\end{algorithmic}
\end{algorithm}

\begin{algorithm}
\caption{check}
\begin{algorithmic}
\Procedure{check}{$progEnv, exam, progState$}
\Ensure True if the program satisfies the examples, false otherwise
\State $target \gets getTargetFunction(progState)$ \Comment{the fn. that must satisfy exam}
\If{$not \, compatible(exam.type, target.type)$}
    \State \textbf{return} $false$
\EndIf

\For {$def \in progState.complete$} \Comment{add complete defn. to env}
	\State $addDef(progEnv, def)$
\EndFor
\For {$pos \in exam.positive$}
	\If{$eval(progEnv, (apply(target.body, pos.in))) \neq pos.out$}
		\State \textbf{return} $false$
	\EndIf
\EndFor
\For {$neg \in exam.negative$}
	\If{$eval(progEnv, (apply(target.body, neg.in))) = neg.out$}
		\State \textbf{return} $false$
	\EndIf
\EndFor
\State \textbf{return} $true$
\EndProcedure
\end{algorithmic}
\end{algorithm}

\par We shall call this algorithm $\mathcal{A}_{linear}$. We will prove $\mathcal{A}_{linear}$ is sound and complete. One of the main difficulties is to show that our approach to typing works correctly even thought we go about it in a top-down manner (as opposed to the usual bottom-up one): there are functions whose types we do not fully know yet; this means that we work with partial information a lot of the time, and by acquiring more and more constraints on the types we reach the final ones. In contrast, a normal type inference algorithm knows the full types of a program's function.

\begin{lemma}
Let $\mathcal{P}_{\preccurlyeq} = \{\, P \, \,| (P_{empty}, \Gamma_{empty}) \preccurlyeq^* (P, \Gamma) \land \Gamma $ is a consistent typing environment for $P  \,\}$. Then we have $\mathcal{P}_{\preccurlyeq} = \mathcal{P}_{BK}$ (where $BK$ only contains linear templates).
\end{lemma}

\begin{proof}
We do so by double inclusion.
\begin{itemize}
\item[$\subseteq$]:
We need to show that if $P \in \mathcal{P}_{\preccurlyeq}$ then $P \in \mathcal{P}_{BK}$. This follows from the rules of $\preccurlyeq$, since the programs in $\mathcal{P}_{\preccurlyeq}$:
\begin{itemize}
\item only use functions from $BK_I$ and $BK_F$ to fill holes, and templates from $BK_T$ when assigning templates;
\item are not cyclical (see the second condition in the \textit{specialization step} relation);
\item they are always well typed: this can be proven using a short inductive argument. The empty state is well typed because it contains the empty program. Now, suppose we have reached a program state $(P, \Gamma)$, where $P$ is typable wrt. $\Gamma$ and we have $(P, \Gamma) \preccurlyeq (P', \Gamma')$. If this resulted via the \textit{define} rule, $P'$ must be typable wrt. $\Gamma'$ by how $\Gamma'$ was defined: the to be defined function takes all the constraints that were created from using it in other definitions, and makes sure that the template we apply is compatible with them. Now, if the \textit{specialize} rule was used, we only replace a hole when the function we fill it with has a unifiable type with the type that can be inferred for the hole (we essentially use the \textit{TApp} rule). So $P'$ must be typable.
\end{itemize}
The highlighted points correspond to the conditions a program must satisfy to be in $\mathcal{P}_{BK}$.

\item[$\supseteq$]:  we must prove that if $P \in \mathcal{P}_{BK}$ then $P \in \mathcal{P}_{\preccurlyeq}$. Since $P \in \mathcal{P}_{BK}$, we know it is typable wrt. a typing environment $\Gamma$. Now, observe that if we pick a function name that occurs in $P$, and replace that with a hole, the resulting program must be typable (by the \textit{THole} rule) wrt. a typing environment (1). Also, when the body of a function is a template which has no filled holes, we can delete that definition and still have a typable program (the function's name will still be in the typing environment, but its type becomes more general because we removed a constraint) (2). Those two actions can be seen as the reverse of the \textit{define} and \textit{specialize} rules. Now, consider the following process: pick a complete function; replace each function name with a hole, until we are left with a template that contains only holes; then delete that function; repeat the process again. It is clear that this process will eventually lead to the empty program, since at each step we either delete a function, or get closer to deleting a function. Using (1) and (2), we know that all the intermediary programs will be typable, wrt. some typing environments. Furthermore, each of those programs will be acyclical and will only contain information from $BK_I$ and $BK_F$. This means that our process creates a $\preccurlyeq$-path in reverse, say $(P, \Gamma) \succcurlyeq (P_1, \Gamma_1) \succcurlyeq \dots \succcurlyeq (P_N, \Gamma_N) \succcurlyeq (P_{empty}, \Gamma_{empty})$, for some states $(P_i, \Gamma_i)$, since what we have essentially done is apply the two rules of $\preccurlyeq$ in reverse. We do need to make the following note though: this reverse constructions works when we consider only linear templates, because those guarantee that no undesirable side effects occur when a hole is filled, so we can construct a $\preccurlyeq$-path starting at either end. Hence, we have that $P \in \mathcal{P}_{\preccurlyeq}$.
\end{itemize}
From the above cases, we have the conclusion.
\end{proof}

\begin{theorem}
$\mathcal{A}_{linear}$ is sound and complete when considering only linear templates.
\end{theorem}
\begin{proof}
By lemma 4.1 and since \textit{expand} implements the $\preccurlyeq$ relation (and hence synthesizes all programs in $\mathcal{P}_{\preccurlyeq}$), we have that \textit{expand} produces all programs in $\mathcal{P}_{BK}$ (note that it does not matter that \textit{expand} uses the rules in a specific order, since the proof of lemma 4.1 did not consider a specific application order). Since $\mathcal{S}_{BK, \mapsto^+, \mapsto^-} \subset \mathcal{P}_{BK}$ and \textit{check} is used on all programs synthesized by \textit{expand}, we are certain that we will be able to synthesize all functions in $\mathcal{S}_{BK, \mapsto^+, \mapsto^-}$ (1). Furthermore, since a program is the output of the algorithm only if it satisfies the examples (enforced by \textit{check}), this means that the programs we synthesize are indeed solutions (2). From (1) and (2) we have that the algorithm is sound and complete.
\end{proof}

\par Hence, From theorem 4.1 we have that $\mathcal{A}_{linear}$ solves the program synthesis problem.
\par Motivated by question \textbf{Q2} from section 1.2, in this subsection we have focused on creating an algorithm that employs effective and early type based pruning. We have seen that if linear templates are used, the pruning technique is a ``natural" extension of what a normal type inference system would do. While only using linear templates might seem restrictive, we note that templates such as \textit{map}, \textit{filter}, \textit{fold} (with the type of the base element a base type, so that we don't have shared type variables) are all linear templates, and they can express a wide range of problems.
\par We make one last observation here. One can notice that \textit{expand} does not care about the types of the examples being compatible with the type of the incomplete programs' target type (this is done in \textit{check}). To increase the amount of pruning, we can make \textit{expand} discard those programs whose target type does not agree with the type of the examples. We observe that this does not break the completeness of the algorithm, since the programs we wish to synthesize are in the solution space, and hence must have their target's type compatible with the type of the examples.

\subsection{Consequences of adding branching templates}
We now investigate what effect branching templates have on $\mathcal{A}_{linear}$. We first prove the next important result.
\begin{theorem}
$\mathcal{A}_{linear}$ is no longer complete when its input contains branching function templates.
\end{theorem}

\begin{proof}
We will show that $\mathcal{A}_{linear}$ is unable to synthesize the following problem: given a list of lists, reverse all inner lists as well as the outer list. Suppose that the available function templates are \textit{map}, \textit{composition} and \textit{identity} and that the only background function is \textit{reverse}. Given those, we should be able to infer the following program: \\
\indent $target = gen1 \, . \, gen2$ \\
\indent $gen2 = map \, gen1$ \\
\indent $gen1 = reverse$ \\
Since we have fixed the order in which we define functions and fill the holes (see the description of \textit{expand}), the following derivation sequence will happen on the path to finding the solution.
\begin{itemize}
\item[1.] $target = \fbox{1} \, . \, \fbox{2} $ (we have no holes to fill, define the \textit{target} function)
\item[2-3.] $target = gen1 \, . \, gen2$ (before inventing any other function fill the 2 existing holes)
\item[4.] $target = gen1 \, . \, gen2$, $gen2 = map \, \fbox{3}$ (all holes filled, define $gen2$)
\item[5.] $target = gen1 \, . \, gen2$, $gen2 = map \, gen1$ (before inventing any other function fill the existing hole)
\item[] ...
\end{itemize}
After step 1, because $(.)$ is applied to the two holes, we infer that the types of $\fbox{1} $ and $ \fbox{2}$ must be of the form $b \rightarrow c$ and $a \rightarrow b$, respectively. In steps 2-3, since we invent the functions \textit{gen1} and \textit{gen2}, their types will be $b \rightarrow c$ and $a \rightarrow b$ (since they don't have a definition, and are also not used elsewhere, they just take the types of the holes they filled). In step 4, the type of \textit{gen2} will change: we must then unify \textit{gen2}'s type with that of the \textit{map} template ($[d] \rightarrow [e]$). After this unification takes place, \textit{gen2}'s type will be $[d] \rightarrow [e]$, the type that can be inferred for $\fbox{3}$ is $d \rightarrow e$, but $gen1$'s type will also change (since it shares a type variable with $gen2$'s type) to $[e] \rightarrow c$. Now, in step 5, we need to unify \textit{gen1}'s type with the type of $\fbox{3}$, and we will get $[e] \rightarrow e$ for \textit{gen1}. Now, since we know that the type of \textit{reverse} is of the form $[f] \rightarrow [f]$, it is clear that the definition $gen1 = reverse$ is impossible, since \textit{gen1}'s type is incompatible with \textit{reverse}'s type. We conclude the program we considered can not be synthesized and hence the algorithm is no longer complete.
\end{proof}

\par The problem here is caused by combining general polymorphism with branching templates. The fact that we are trying to type programs progressively, in top-down manner, creates the possibility of constraining types too early. An example can be seen in the previous proof: \textit{gen1} and \textit{gen2} initially shared a type variable, because we eagerly adjusted their types to fit with the type of the composition template; this, in turn, lead to the type of \textit{gen1} being ultimately incompatible with the type of \textit{reverse}. Normal type inference can deal with this problem because it fully knows the types of all functions (whereas we are progressively approaching those final types).

\par Our attempt to solving this involved transforming the branching templates into linear templates and generating unification constraints on relevant type variables to correctly deal with polymorphism. For example, take the composition template, whose ``branching" type is $(b \rightarrow c) \rightarrow (a \rightarrow b) \rightarrow (a \rightarrow c)$; the new ``linear" type would be $(b_1 \rightarrow c) \rightarrow (a \rightarrow b_0) \rightarrow (a \rightarrow c)$, with the constraint that $b_0$ and $b_1$ are unifiable. The idea here was that this gives us some ``wiggle room" when deciding the types (and hence potentially fixes the problem introduced by polymorphism). However, we were unable to use this approach to develop an algorithm we were certain was sound and complete. After this attempt, we managed to find two papers that talk about typing a lambda like language with first class contexts, by Hashimoto and Ohori \cite{hashimoto} and by Hashimoto \cite{hashimoto1}. The former paper provides a theoretical basis by creating a typed calculus for contexts, while the latter develops an ML-like programming language that supports contexts and furthermore provides a sound and complete inference system for it. This suggests that there might be a way to have meaningful type pruning when branching templates are allowed, but we will reserve exploring this avenue for future work.

\subsection{An algorithm for branching templates}
\par We now propose an algorithm that is complete and sound when considering branching templates, which we will call $\mathcal{A}_{branching}$. Motivated by the observations shown in the last subsection, we will completely disregard early type pruning for the purposes of this algorithm. $\mathcal{A}_{branching}$ is similar in structure to $\mathcal{A}_{linear}$, but it defers type checking until after we synthesize a complete program (no early type pruning):
\begin{itemize}
\item \textit{progSearch} will remain the same.
\item \textit{expand} will follow the same filling and defining strategy as before and make sure the synthesized programs are acyclical, but will now completely disregard anything type related.
\item  \textit{check} must now have an additional step, checking whether the program is indeed typable (using a ``normal" inference algorithm based on the typing rules of the language) before checking for example satisfiability.
\end{itemize}

It is easy to see that this algorithm is sound (because of the extra check in \textit{check}) and complete (\textit{expand} produces every possible program that is not cyclical, which is clearly a superset of the solution space), and hence solves the program synthesis problem. For the experiments involving reuse in the next chapter we shall use an implementation of this algorithm, because, as we will see in the next chapter, branching templates are important for effective function reuse.
\chapter{Experiments and results}
\indent \indent In this chapter, through experimentation, we will attempt to answer questions \textbf{Q1}, \textbf{Q3} and \textbf{Q4} from section 1.2, which we reiterate:
\begin{itemize}
\item[\textbf{Q1}] Can function reuse improve learning performance (find programs faster)?
\item[\textbf{Q3}] What impact does the grammar of the synthesized programs have on function reuse?
\item[\textbf{Q4}] What classes of problems benefit from it?
\end{itemize}
\par The implementation we shall use was written in Haskell and closely follows $\mathcal{A}_{branching}$, whose outline can be found in Appendix A. We focus on this algorithm's implementation because, as we will see in section 5.2, using only linear templates makes it almost impossible to reuse functions.

\section{Experiments}
\par \par We begin by showcasing the experiments we conducted in order to answer \textbf{Q1}. For simplicity (and since they have enough expressive power for our purpose), we will use three templates: \textit{map}, \textit{filter} and \textit{composition}. The results of the experiments are summarised in tables 5.1 (output programs) and 5.2 (learning times); the times shown are of the form [mean over 5 repetitions] $\pm$ standard error.

\subsection{addN}
We begin with a simple yet interesting example.
\begin{itemize}
\item[] \textbf{Problem}: Given a number, add N to it.
\item[] \textbf{Method}: We will be considering this problem for $N \in \{4, 5, \dots, 10\}$. The only background function we use is \textit{add1}. Example wise, we will use 2 positive examples of the form $x \rightarrow^+ x + N$ and 2 negative examples of the form $x \rightarrow^- x + M$, with $M \neq N$.
\item[] \textbf{Results}: Figure 5.1 plots the mean of the learning times for different values of $N$ (5 repetitions). As we can see, function reuse is vital here: by creating a function that adds two, we can reuse it to create a function that adds 4, and so on; this means that there is a logarithmic improvement in program size from the no reuse variant, as can be seen in table 5.1 (which shows the solution for \textit{add8}), which in turn leads to an increase in performance, as can be seen in table 5.2. Something to note is that for $N=16$, if reuse is used, the solution is found in under a second, whereas if reuse is not used no solution is found even after 10 minutes. The result here suggest that the answer to question \textbf{Q1} is yes.
\end{itemize}

\begin{figure}
\begin{center}
\begin{tikzpicture}
\begin{axis}[xlabel={N}, ylabel={learning time (sec)}, legend style={at={(0.03,0.5)},anchor=south west}]
    \addplot+[error bars/.cd,y dir=both,y explicit]
    coordinates { (4,0.006636)+-(0.0003756, 0.0003756) (5, 0.010378)+-(0.0003283, 0.0003283) (6,0.014052)+-(0.00089, 0.00089) (7, 0.045016)+-(0.0004639, 0.0004639) (8,0.013348)+-(0.0005713, 0.0005713) (9, 0.046728)+-(0.0015674, 0.0015674) (10,0.048108)+-(0.0004255, 0.0004255) };
    \addplot+[error bars/.cd,y dir=both,y explicit]
    coordinates { (4,0.009872)+-(0.000708, 0.000708) (5, 0.021986)+-(0.0046309, 0.0046309) (6,0.046578)+-(0.0016173, 0.0016173) (7, 0.211948)+-(0.0025156, 0.0025156) (8,1.186)+-(0.0147, 0.0147) (9, 6.684)+-(0.1342, 0.1342) (10,43.246)+-(1.4449, 1.4449) };
\addlegendentry{Reuse}
\addlegendentry{No reuse}
\end{axis}
\end{tikzpicture}
\end{center}
\caption{Learning times for \textit{addN}}
\end{figure}
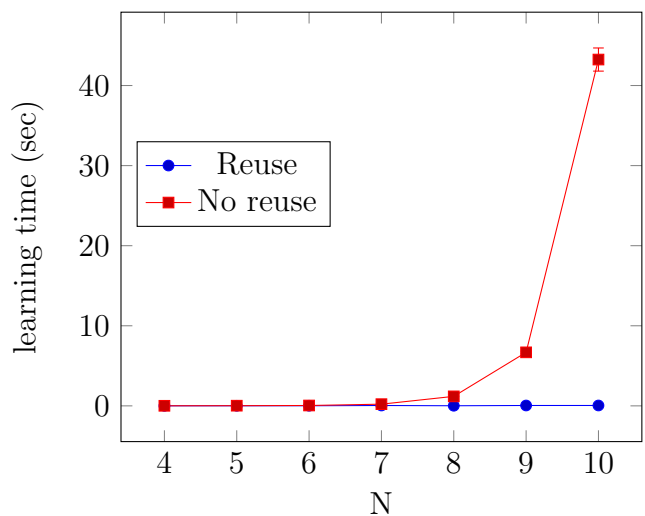

\subsection{filterUpNum}
\begin{itemize}
\item[] \textbf{Problem}: Given a list of characters, remove all upper case and numeric elements from it.
\item[] \textbf{Method}: We use the following background functions:  \textit{isUpper}, \textit{isAlpha}, \textit{isNum}, \textit{not} and will use 2 positive and 2 negative examples.
\item[] \textbf{Results}:  As can be seen in the table 5.1, this is a problem where only function invention suffices, and reuse shows no improvement wrt. program size. However, this is a good example that shows that, for programs which have a reasonably small number of functions, reuse does not introduce too much computational overhead: in our case it doubles the execution time, but this is not noticeable since both execution times are under half a second.
\end{itemize}

\subsection{addRevFilter}
\begin{itemize}
\item[] \textbf{Problem}: Given a list of integers, add 4 to all elements, filter out the resulting even elements and reverse the new list.
\item[] \textbf{Method}: The background functions used are: \textit{add1}, \textit{add2}, \textit{isOdd}, \textit{reverse}; we use 2 positive examples and 2 negative examples.
\item[] \textbf{Results}:  Again, in this case function reuse does not lead to a shorter solution. However, in this case there is a noticeable increase in the execution time when reuse is used, from around 2 seconds to around 9 seconds. Another interesting observation here is that, while both programs in table 5.1 are the smallest such programs (given our BK), the one that employs reuse is actually less efficient than the one that only uses invention: one maps \textit{ add2} twice over the list, whereas the other one creates a function that adds 4 and then maps that over the list once. The result here, together with the result in \textit{filterUpNum} suggest the following about \textbf{Q1}: while function reuse could be very helpful in some situations, sometimes it will not help in finding a shorter solution; furthermore, while in some cases the computational overhead is not too noticeable, in others the overhead will be quite sizeable.
\end{itemize}

\subsection{maze}
\begin{itemize}
\item[] \textbf{Problem}: Given a maze that can have blocked cells, a robot must find its way from the start coordinate to the end coordinate.
\item[] \textbf{Method}: The background functions used represent the basic movements of the robot: \textit{mRight}, \textit{mLeft}, \textit{mDown}, \textit{mUp} (if the robot tries to move out of the maze, ignore that move). The mazes we will consider will be 4x4, 6x6 and 8x8; the start will always be cell (0, 0) and the goals (3, 3), (5, 5) and (7, 7), respectively. We will use one positive example and no negative examples (no need for them in such a problem).
\item[] \textbf{Results}: Reuse has a dramatic effect on the learning times, as can be seen in table 5.2 (for the 4x4 problem). Interesting here are the 6x6 and 8x8 variants, since when enabling reuse we managed to find solutions for both in under 10 seconds, but when reuse was not employed, the system was not able to produce results even after 10 minutes. The result here enforces our previous assertion about question \textbf{Q1}: when reuse is applicable, it can make a big difference.
\end{itemize}

\subsection{droplasts}
\begin{itemize}
\item[] \textbf{Problem}: Given a list of lists, remove the last element of the outer list as well as the last elements of the inner lists.
\item[] \textbf{Method}: The background functions we use are \textit{reverse}, \textit{tail}, \textit{addOne}, \textit{addTwo}, \textit{isOdd}, \textit{id} (the latter 4 functions are noise, to put stress on the system).
\item[] \textbf{Results}: From the formulation, we can get a sense that \textit{tail} combined with \textit{reverse} will represent the building block for the solution, since intuitively this operation would need to be performed on both the outer list as well as on the inner lists. Indeed, the solution that uses function reuse is both shorter and it is found much faster than the no reuse variant. As we can see, reuse has had a major impact here, drastically reducing the computation time (as can be seen in table 5.2). 
\par Curious about how the system (using reuse) would behave when varying the number of background functions, we have conducted a few more experiments to test this. To make it challenging, we have only retained \textit{reverse} and \textit{tail}, and all the functions we added were the identity functions (with different names), so even if type based pruning would be used, it would not really make a difference. The results can be seen in figure 5.2 (we plotted the means of 3 executions $\pm$ the standard error). The results here enforce our previous assertion about question \textbf{Q1}, solidifying our belief that reuse is indeed useful, while also showing that the system behaves respectably when increasing the number of background functions.
\end{itemize}

\begin{figure}
\begin{center}
\begin{tikzpicture}
\begin{axis}[xlabel={\# of BK functions used}, ylabel={learning time (sec)}]
    \addplot+[error bars/.cd,y dir=both,y explicit]
    coordinates { (2,0.1492433) +- (0.0019439, 0.0019439) (3,0.4220267) +- (0.0070662, 0.0070662) (4,1.04)+-(0.0153, 0.0153) (5, 1.9967) +- (0.012, 0.012) (6,3.5467)+-(0.0186, 0.0186) (7,6.02)+- (0.0173, 0.0173) (8,9.8367)+-(0.2382, 0.2382) (9,15.8533)+-(0.98, 0.98) (10,22.4633)+-(0.9607, 0.9607) };
\end{axis}
\end{tikzpicture}
\end{center}
\caption{Learning times for \textit{droplasts} (with reuse)}
\end{figure}
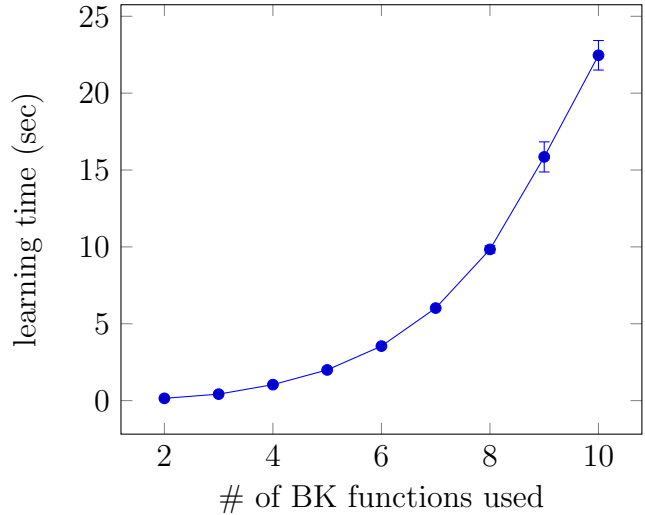

\begin{table}[h!]
\begin{center}
\begin{tabular}{|c|c|c|}
\hline
Problem & Reuse + Invention & Only invention \\

\hline
add8 &
\begin{tabular}{@{}c@{}}  g3 = add1.add1 \\ gen2 = g3.g3 \\ target = g2.g2   \end{tabular} &
 \begin{tabular}{@{}c@{}}
g3 = add1.add1 \\
g5 = add1.add1 \\
g7 = add1.add1 \\
g6 = add1.add1 \\
g4 = g6.g7 \\
g2 = g4.g5 \\
target = g2.g3   \end{tabular} \\
\hline

filterUpNum &
\begin{tabular}{@{}c@{}}
g3 = filter isAlpha \\
g4 = not.isUpper \\
g2 = filter g4 \\
target = g2.g3 \\
\end{tabular} &
same as R + I \\
\hline

addRevFilter &
\begin{tabular}{@{}c@{}}
g5 = g4.reverse \\
g4 = map add2 \\
g3 = g4.g5 \\
g2 = filter isOdd \\
target = g2.g3 \\
\end{tabular} &
\begin{tabular}{@{}c@{}}
g5 = add2.add2 \\
g4 = map g5 \\
g3 = g4.reverse \\
g2 = filter isOdd \\
target = g2.g3 \\
\end{tabular} \\

\hline
maze(4x4) &
\begin{tabular}{@{}c@{}}
g3 = mRight.mUp \\
g2 = g3.g3 \\
target = g2.g3 \\
\end{tabular} &
\begin{tabular}{@{}c@{}}
g3 = mRight.mUp \\
g5 = mRight.mRight \\
g4 = mUp.mUp \\
g2 = g4.g5 \\
target = g2.g3
\end{tabular} \\

\hline
dropLasts &
\begin{tabular}{@{}c@{}}
g4 = reverse.tail \\
g3 = g4.reverse \\
g2 = map g3 \\
target = g2.g3 \\
\end{tabular} &
\begin{tabular}{@{}c@{}}
g6 = reverse.tail \\
g3 = g6.reverse \\
g5 = reverse.tail \\
g4 = g5 reverse \\
g2 = map g4 \\
target = g2.g3 \\
\end{tabular} \\

\hline
\end{tabular}
\end{center}
\caption{Programs output for the experimental problems}
\end{table}

\begin{table}[h!]
\begin{center}
\begin{tabular}{ |c|c|c| }
\hline
Problem & Reuse + Invention & Only invention \\

\hline
add8 &
13.34 ms $\pm$ 0.57&
1.18 sec $\pm$ 0.01\\
\hline

filterUpNum &
338.29 ms $\pm$ 11.16 &
153.90 ms $\pm$ 4.57 \\
\hline

addRevFilter &
9.11 sec $\pm$ 0.06 &
1.97 sec $\pm$ 0.01 \\

\hline
maze(4x4) &
67.50 ms $\pm$ 4.85&
5.37 sec $\pm$ 0.05\\

\hline
dropLasts &
1.84 sec $\pm$ 0.01  &
252.24 sec $\pm$ 7.16 \\

\hline
\end{tabular}
\end{center}
\caption{Learning times for the experimental problems}
\end{table}

\section{On function reuse and function templates}
\indent \indent After attempting to answer \textbf{Q1} in the previous subsections, we now consider \textbf{Q3} and \textbf{Q4}. As previously discussed, function reuse does not come without a cost: in some cases it negatively affects the execution time. Furthermore, it is clear that not all programs take advantage of function reuse, so an answer to question \textbf{Q4} is quite important.

\par We have been able to distinguish two classes of problems that benefit from function reuse: problems that involve repetitive tasks (especially planning in the AI context) and problems that involve operations on nested structures. For the former class, we can give the \textit{maze} problem as an example. In that case, function reuse lead to the creation of a function that was equivalent to \textit{moveRight} then \textit{moveUp}, which helped reach shorter solutions because the robot used this combination of moves  frequently. For the latter class, \textit{droplasts} is a perfect illustration. The solution acts on both the inner and outer lists, which is a good indication that the operation might be repetitive and hence benefit from the reuse of functions. We note that both classes of programs we presented contain quite a lot of programs (lots of AI planning tasks can be encoded in our language and tasks that act on nested structures are common), which is a good indication that reuse is applicable and can make a difference in practical applications.

\par Another interesting point (raised by \textbf{Q3}) is how the presence of various function templates (which induce the grammar of our programs) affects function reuse, and whether the partition we have presented in the previous chapter plays a part in this. If we think about the graph the \textit{uses} relation induces on the invented functions (call it a functional dependency graph, or FDG), the programs our algorithms synthesize have acyclic FDGs, because we never introduce cyclical definitions (see definition 4.2). Now, the fact that the types of the holes do not share type variables for linear templates suggests that in practice the majority of those templates are actually likely to have one single hole. If this is the case, this means they create linear FDGs, which make function reuse impossible (otherwise, the FGD would be cyclic). Indeed, we can make the following remark: to enhance function reuse, branching templates should always be used. In particular, \textit{composition} and other similar branching templates that encapsulate the idea of chaining computations are very effective: they create branches in the FDG, and a function invented on such a branch could be reused on another one.

\chapter{Conclusions}
\section{Summary}

\indent \indent This project was motivated by the fact that, to the best of our knowledge, invention and in particular function reuse has not been properly researched in the context of inductive functional programming.
\par In chapter 3, we formalized the program synthesis problem, which provided us with the language necessary to create algorithms. Chapter 4 represents an important part of the project, since that is where we have presented two approaches to solving the program synthesis problem, namely $\mathcal{A}_{linear}$ (which works with a specific type of background knowledge) and $\mathcal{A}_{branching}$ (which works on general background knowledge). An interesting result we found was that the form of type pruning $\mathcal{A}_{linear}$ uses, which relies on a normal type inference process (extended in a natural way to work with contexts), breaks its completeness if general background knowledge is used: hence, we observed that type pruning is not a trivial task when synthesizing modular programs. In chapter 5 we have relied on the implementations of $\mathcal{A}_{branching}$ to show a variety of situations where function reuse is important: examples such as \textit{droplasts}, \textit{add8} and \textit{maze} showed how crucial it can be. We have also distinguished two broad classes of programs that will generally benefit from function reuse and discussed the impact the used background knowledge has on reuse.
\par Overall, we have seen that there is value in exploring the ideas of modular programs and function reuse, and believe that this project can serve as the base for future work in this direction.
\section{Reflections, limitations and future work}
\indent \indent The project is limited in a few ways and can be further enhanced. One major limitation is the lack of any form of pruning for $\mathcal{A}_{branching}$ (the algorithm that works with linear templates benefits from type based pruning). As we have stated in chapter 4, a possible way to overcome the problems of typing with contexts could be solved by attempting to create a type system and type inference system similar to the ones described in \cite{hashimoto} and \cite{hashimoto1}. Furthermore, a possible extension of the project could examine the benefits of pruning programs through example propagation, in a similar way to how $\lambda^2$ does it \cite{lambdasq}. An interesting point to explore is whether branching templates would hinder this approach to pruning in any way (more specifically, whether templates such as \textit{composition} would prevent any such pruning to be done before the program is complete). Another avenue to explore would be to see whether there are other major classes of programs that benefit from function reuse, specifically problems related to AI and game playing.

%now enable appendix numbering format and include any appendices
\appendix
\chapter{Implementation details}

We briefly give some implementation details for the algorithm $\mathcal{A}_{branching}$. The implementation can be found at https://github.com/reusefunctional/reusefunctional, which also contains details on how to run the system.

\section{Language implementation}
The language that we have presented in chapter 4 is very similar to the language \textit{Fun}, presented in the Oxford PoPL course \cite{popl}. Hence, we have used the parser and lexer, together with parts of the interpreter for our implementation, but those have been extended in multiple ways. We have added types to the language and added a type inference system (which can be found in the files \textit{Types.hs} and \textit{Infer.hs}). The inference system follows classical algorithms, and a similar implementation by Stephen Diehl can be found in \cite{diehl}, which we have used as a guide. To support the synthesis process inside the language, we have added three constructs to the language presented in chapter 4.
\par Listing A.1 shows the test file used for the \textit{add} problem mentioned in chapter 5, which highlights most features of the language. Note that the first three functions represent the implementation of the higher order functions used in the templates (this is part of our idea that templates should be easy to modify) and the fourth is a background function (hence the \textit{BK_} prefix).

\begin{lstlisting}[frame=single, caption=add8 test file]
val comp(f, g) = lambda (x) f(g(x)) ;;

rec map(f) = lambda (xs)
    (if xs = nil
    then nil
    else f(head(xs)):map(f)(tail(xs))) ;;

rec filter(p) = lambda (xs)
    (if xs = nil
    then nil
    else 
        if (p(head(xs)))
        then head(xs):filter(p)(tail(xs))
        else filter(p)(tail(xs))) ;;

val BK_addOne(x) = x + 1 ;;

NEx (1) => 2 ;;
NEx (3) => 5 ;;
PEx (1) => 9 ;;
PEx (7) => 15 ;;
Synthesize (Int) => Int;;
\end{lstlisting}

\section{Algorithms implementation}
The implementation of the searching algorithm closely follows the algorithm described in chapter 4 and can be found in the file \textit{Search.hs}. The \textit{uses} relation and the cycle check is done by creating a dependency graph and checking for cycles when adding edges during the creation of declarations (\textit{DepGraph}). The \textit{check} function has been implemented with the help of an interpreter for our language (\textit{Interpreter.hs}). To have the induced functions in the right order when defining them for the purpose of checking the examples, we use the topological ordering of the dependency graph's nodes (which denote the functions) to make sure a function can only be defined once all the functions that appear in its body are also defined.

%next line adds the Bibliography to the contents page
\addcontentsline{toc}{chapter}{Bibliography}
%uncomment next line to change bibliography name to references
\renewcommand{\bibname}{References}
\bibliography{Bibliography} %use a bibtex bibliography file refs.bib
\bibliographystyle{plain}  %use the plain bibliography style

\end{document}